\documentclass[11pt]{article}

\usepackage[letterpaper,margin=1in]{geometry}
\usepackage[T1]{fontenc}
\usepackage[utf8]{inputenc}
\usepackage{amsmath,amsthm,mathtools} 
\usepackage{newtxtext,newtxmath}
\usepackage{booktabs}
\usepackage{array}
\usepackage{enumitem}
\usepackage{tikz}
\usetikzlibrary{arrows.meta,positioning,calc}
\usepackage{microtype}
\usepackage{natbib}
\usepackage[colorlinks=true,linkcolor=blue,citecolor=blue,urlcolor=blue]{hyperref}
\allowdisplaybreaks

\newtheorem{theorem}{Theorem}[section]
\newtheorem{lemma}[theorem]{Lemma}
\newtheorem{proposition}[theorem]{Proposition}
\newtheorem{corollary}[theorem]{Corollary}
\theoremstyle{definition}
\newtheorem{definition}[theorem]{Definition}
\newtheorem{fact}[theorem]{Fact}
\newtheorem*{theoreminf}{Theorem}
\theoremstyle{remark}
\newtheorem*{remark}{Remark}

\title{Internalizing Chain-of-thought into Looped Transformers}

\author{%
  Kehang Zhou \qquad Pinyan Lu\thanks{Corresponding author.}\\[0.5em]
  {\normalsize Institute for Theoretical Computer Science}\\
  {\normalsize Shanghai University of Finance and Economics}\\[0.3em]
  {\normalsize\texttt{kehang.zhou@icloud.com} \quad \texttt{lu.pinyan@mail.shufe.edu.cn}}%
}
\date{}

\begin{document}
\maketitle
\begin{abstract}

Chain-of-thought (CoT) enables constant-bit-size transformers to perform general-purpose computation through generated tokens \citep{liwang2025constant}, whereas looped transformers carry out the same computation by recurrence in hidden states (\emph{latent thought}). We construct a family of constant-bit-size looped transformers that interpolates between CoT and latent thought. For a prescribed loop count $L=L(n)$ and a fixed integer $k\ge1$ satisfying $s(n)^{1/k}=\Omega(L)$, our transformer simulates a time-$t(n)$, space-$s(n)$ Turing machine using a context window of $O(s(n)/L)$ and $O(t(n)s(n)^{1/k}/L)$ generated tokens, reducing both resources by a factor of $L$ relative to the corresponding unlooped queue-machine simulation. The construction also holds with softmax attention at constant bit-size, under the finite-precision kernel specified in this paper.
\end{abstract}

\section{Introduction}\label{sec:intro}

Transformers \citep{vaswani2017attention} can perform general-purpose computation through recurrence. A central notion is \emph{Turing completeness}: the ability to simulate any Turing machine, the standard mathematical model of algorithmic computation. Informally, a Turing-complete model can carry out any computation that a computer program can perform, given sufficient time and memory.

One form of recurrence is autoregressive generation, in which the model repeatedly appends a token to its sequence and uses the resulting sequence to generate the next token. Chain-of-thought (CoT) uses these generated tokens for intermediate computation \citep{wei2022chain}. The context window---the range of token positions accessible to attention---serves as the visible workspace.

\citet{perez2021attention} proved that transformers with autoregressive generation are Turing-complete. The earliest constructions use precision that grows with the input, together with a full-history context that scales with running time \citep{bhattamishra2020computational, merrill2024expressive}. The full-history regime has a provable ceiling: with a polynomial-time next-token map and polynomial context, standard full-history CoT decides only languages in $\mathsf{P}$, hence not all of $\mathsf{PSPACE}$ unless $\mathsf{P}=\mathsf{PSPACE}$ \citep{yang2025pencil}. Sliding-window attention and summarized CoT then bring the visible context down from the running time to the space bound $O(s(n))$ \citep{schuurmans2024autoregressive, yang2025pencil, brosamle2026expressive}, while precision is reduced first to $O(\log(n+t(n)))$ and then to a constant \citep{merrill2024expressive, li2024chain}. \citet{liwang2025constant} complete this line: a constant-bit-size transformer that attends only at fixed relative offsets is Turing-complete. \citet{liwang2025efficient} then bring the cost down to $O(s(n)^{1/k})$ CoT tokens per simulated machine step by increasing the head--layer product.

Looped transformers provide the complementary form of recurrence, in depth: a shared block is applied repeatedly to the hidden state, performing multiple iterations before emitting a token and carrying intermediate results in latent representations rather than an explicit CoT trace \citep{giannou2023looped, yang2024looped, geiping2025scaling}. This internal reasoning is a form of \emph{latent thought}~\citep{saunshi2025reasoning, xusato2026formal}, and \citet{saunshi2025reasoning} show that loop iterations can simulate CoT steps one-for-one. How does moving computation from generated tokens into hidden-state recurrence change the context window and the number of generated tokens required for Turing-machine simulation?

\paragraph{Our contribution.} We construct a family of constant-bit-size looped transformers that interpolates between explicit CoT and latent thought. One shared decoder block is applied $L$ times per generated token, so part of the computation moves from the generated trace into hidden activations. An intermediate queue machine represents the Turing-machine computation, and each of its transitions can be realized either as an explicit CoT step or as one iteration of that block. Relative to the queue-machine simulations of \citet{liwang2025constant, liwang2025efficient}, grouping $L$ simulated transitions into each generated token reduces both the context window and the CoT length by a factor of $L$.

The mechanism is a per-iteration KV cache read through a sliding window. Iteration $l$ writes its result into the iteration-$l$ cache, the cached entries together hold the $\Theta(s(n))$ machine cells, and attention at that iteration reads only $O(s(n)/L)$ positions. A context window of that length therefore completes the simulation of an $s(n)$-space machine. The learned parameters are fixed, machine-dependent constants: they do not grow with $n$ or with $L$. The positional encoding, the context length, the loop count, and the initial cache are instantiated from the resource bounds. The same construction holds with softmax attention under the finite-precision kernel of Section~\ref{sec:softmax}. In the statement below, $k'$ is the number of tapes of the simulated machine and $k$ is the hierarchy parameter of the queue simulation (Section~\ref{sec:models}).
\begin{theoreminf}[informal]
Any $(t(n),s(n))$ time-space bounded Turing machine can be simulated by a constant-bit-size looped transformer (with its input pre-loaded into the initial KV cache as specified in Section~\ref{sec:initialization}) with head--layer product $K=8kk'$ and context window length $O(s(n)/L)$, that takes $O(s(n)^{1/k}/L)$ CoT steps for each simulated TM step, leading to a total CoT length of $O(t(n)\cdot s(n)^{1/k}/L)$, for every loop count $L=L(n)$ and every fixed integer $k\ge1$ with $s(n)^{1/k}=\Omega(L)$. At $L=1$, the bounds recover \citep[Thm~1]{liwang2025efficient}.
\end{theoreminf}

For the constructed simulation, explicit and latent reasoning are not two capabilities but two placements of the same general-purpose computation realized by recurrence. Decreasing $L$ externalizes that computation as generated tokens. Increasing $L$ folds that many consecutive CoT steps into each token as latent iterations. The loop count is the exchange rate between the external trace and the internal computation.

\noindent\textbf{Proof idea.} We fold $L$ consecutive machine steps of an existing CoT simulation into $L$ iterations of the shared block, replacing $L$ generated tokens by one. The obstacle is layer alignment: in ordinary CoT every past token is visible to every attention layer, whereas iteration $l$ of a looped block reads only iteration-$l$ cache entries of earlier positions, so the inputs consumed at a simulated step must sit in the cache at the layer that executes it. Li--Wang's queue-machine simulation admits the required reorganization: each simulated step's queue fronts are cached at the iteration that consumes them. Section~\ref{sec:example} shows the alignment on one queue.

\subsection*{Potential messages for real LLMs}

The same construction suggests several questions for practical reasoning models. Each is a conjecture: none is backed by an experiment, and none is a claim about a deployed model.
\begin{itemize}
    \item Looped transformers have the expressive capacity to internalize a substantial part of a chain of thought into recurrent state. Raising the loop count replaces generated reasoning tokens by computation inside the loop, and general-purpose computation is preserved. The loop count is then a way to separate reasoning computation from the tokens a reader sees.
    \item The number of generated CoT tokens need not measure the serial computation a model performs. In the simulation, one generated token summarizes $L$ consecutive steps, and the intermediate states stay in the recurrent state. Visible length and underlying computation can therefore diverge.
    \item At large $L$ this creates a learning difficulty. The explicit trace is strided: only every $L$-th state is emitted as a token, and the steps in between stay latent. Next-token prediction does not supervise those latent steps, so the computation may be difficult to acquire from pretraining alone.
    \item The decoupling of explicit reasoning trace and latent computation exposes a potential learning challenge at large $L$: the explicit reasoning trace becomes increasingly strided, since only every $L$-th computational state is exposed as a token, while intermediate computation remains latent. Such a trajectory is not directly supervised by the standard next-token prediction objective, which conjecturally could make this form of computation difficult to acquire through pretraining alone.
    \item The construction also suggests how a long-reasoning model might control latent computation, persistent state, and visible tokens separately: take the loop count as an explicit budget for latent computation, and supervise only selected boundary states, so that not every internal step has to be written out as text.
\end{itemize}

Section~\ref{sec:models} fixes the models of computation, and Section~\ref{sec:model} the looped transformer. Section~\ref{sec:hardmax} gives the hardmax construction and its main theorem, with the initialization interface stated explicitly (Section~\ref{sec:initialization}) and a worked example at the end (Section~\ref{sec:example}). Section~\ref{sec:softmax} lifts it to softmax at constant bit-size. Section~\ref{sec:related} discusses related work; Appendix~\ref{app:congruence} constructs congruent queue lengths, Appendix~\ref{app:components} records the transformer components, Appendix~\ref{app:softmax} the exact softmax kernel contract, and Appendix~\ref{app:comparison} the detailed comparison table; Appendix~\ref{app:notation} collects notation.

\section{Models of computation}\label{sec:models}

We simulate Turing machines through an intermediate model, a synchronous multi-queue machine, following \citet{liwang2025constant, liwang2025efficient}.

\subsection{Turing machines and queue machines}

\paragraph{Multi-tape Turing machine.}
A $k'$-tape Turing machine (TM; $k'$ is the tape count) is a tuple $\langle \Sigma, Q, \delta \rangle$ where
\begin{itemize}
\item $\Sigma$ is a finite tape alphabet, including $0$, $1$, and a blank symbol $\#$.
\item $Q$ is a finite set of states, including a start state $q_{\mathrm{start}}$ and halting states $q_{\mathrm{acc}}, q_{\mathrm{rej}}$.
\item $\delta : Q \times \Sigma^{k'} \to Q \times \Sigma^{k'-1} \times (\{\mathrm{Left}, \mathrm{Stay}, \mathrm{Right}\})^{k'}$ is a transition function.
\end{itemize}
The first tape is a read-only input tape, initially holding a finite $\{0,1\}$-string of input $x$; the remaining $k'{-}1$ tapes start empty. A TM is $(t(n), s(n))$ time-space bounded if, on every input of length $n$, it halts within $t(n)$ steps and visits at most $s(n)$ distinct cells across the $k'$ tapes. Throughout the paper, $k'$ always denotes the tape count, leaving $k$ free for the hierarchy exponent of Fact~\ref{fact:queue} below.

\paragraph{Synchronous multi-queue machine.}
The intermediate model is a synchronous multi-queue machine: the multi-queue form of the Post machine \citep{post1936finite} in the variant of \citet[\S2.2]{liwang2025efficient}. Each queue is a first-in, first-out tape, read at the front and written at the rear. At every step every queue pops its front symbol and appends one symbol, so the queues move in lockstep. There is no separate stay-or-move choice; every front advances by one cell.

\begin{itemize}
\item $\Sigma$ is a finite alphabet, containing $0$, $1$, and the blank $\#$.
\item $Q$ is a finite state set, containing the start state $q_{\mathrm{start}}$ and the halting states $q_{\mathrm{acc}}, q_{\mathrm{rej}}$.
\item $\delta : Q \times \Sigma^K \to Q \times \Sigma^K$ reads the $K$ front symbols and returns the next state together with the $K$ symbols appended at the rears.
\end{itemize}

Initially the input queue holds the input string and the other $K{-}1$ queues hold only blanks. The space of a queue machine is the maximum number of symbols stored in its queues during the computation, over inputs of length $n$. In the rest of the paper, $s(n)$ is only the simulated Turing machine's space bound; the queue machine uses $O(s(n))$ space.

We take queue $0$ to be the input queue, and it is the longest. Every halting state is absorbing: for $q \in \{q_{\mathrm{acc}}, q_{\mathrm{rej}}\}$ and every $K$-tuple $\sigma \in \Sigma^K$, set $\delta(q, \sigma) = (q, \#, \dots, \#)$. The convention makes every run an infinite log. It adds at most $L{-}1$ extra transitions in the final pass; the halting verdict is unchanged (Proposition~\ref{prop:queue}).

\begin{definition}[K-queue execution log]\label{def:log}
Let $M = (\Sigma, Q, \delta)$ be a synchronous $K$-queue machine with fixed queue lengths $\ell_0, \dots, \ell_{K-1}$, on input $x = x_0 \cdots x_{n-1}$ of length $n \le \ell_0$, and let $s_{\max}=\max_r\ell_r$. Its execution log is $(\sigma_m, q_m)_{m \ge 0}$, where $\sigma_m = (\sigma^{(0)}_m, \dots, \sigma^{(K-1)}_m)$. The entries with $m \le s_{\max}-1$ are given directly: the input feed, for $0 \le m \le n-1$,
\[
(\sigma_m,\, q_m) = (x_m,\ \#^{K-1},\ q_{\mathrm{start}}),
\]
and the padding, for $n \le m \le s_{\max}-1$,
\[
(\sigma_m,\, q_m) = (\#^{K},\ q_{\mathrm{start}}).
\]
For $m \ge s_{\max}-1$,
\[
(q_{m+1},\, \sigma_{m+1}) = \delta\big(q_m,\ \sigma^{(0)}_{m-\ell_0+1},\ \dots,\ \sigma^{(K-1)}_{m-\ell_{K-1}+1}\big),
\]
every front index $m-\ell_r+1 \ge s_{\max}-\ell_r \ge 0$ being given or previously produced; the absorbing convention makes the log infinite. The remaining run length after the cache is filled is
\[
\tau(n)=\max\big(1,\ \min\{m \mid q_m \in \{q_{\mathrm{acc}}, q_{\mathrm{rej}}\}\}-(s_{\max}-1)\big).
\]
\end{definition}

The step-$m$ front of queue $r$ is the symbol appended at step $m-\ell_r+1$; $\tau(n)$ counts the transitions from step $s_{\max}-1$, the first step whose fronts are available, to the first halting state.

The conversion from Turing machines to this model is cited, not re-derived:

\begin{fact}[{Queue Turing machines are Turing-complete \citep[Thm.~2, Step~I]{liwang2025efficient}}]\label{fact:queue}
For every fixed integer $k\ge1$, a $t(n)$-time, $s(n)$-space $k'$-tape Turing machine converts to a synchronous $K=6kk'$-queue machine running in $O(t(n)s(n)^{1/k})$ steps and $O(k's(n))$ total queue space. Each tape becomes two stacks and each stack $k$ levels of queues. Writing $B=\lceil s(n)^{1/k}\rceil$, level $g$ ($1\le g\le k$) uses two content-half queues of length $B^g$ and one buffer of length $2B^g$; consecutive half lengths divide one another, and the cost analysis is in terms of these lengths \citep[\S3]{liwang2025efficient}.
\end{fact}

Two remarks on scope. First, the level structure and the cost bounds are what we cite; the congruent base $b$, the single length per level, the buffer split, and the resource bookkeeping are argued in Appendix~\ref{app:congruence}. Second, $k$ is fixed before instantiation (it sets the head count $K=8kk'$) and $L$ is prescribed afterwards, subject to $s(n)^{1/k}=\Omega(L)$.

\section{The model: looped transformers with a per-iteration cache}\label{sec:model}

\subsection{Decoder-only transformers with sliding-window attention}

A decoder-only transformer \citep{vaswani2017attention} is a parameterized next-token map $\mathcal{V}^{*} \to \mathcal{V}$. The forward pass for position $i$ sets $h_i^{(0)} = \mathrm{emb}(v_i)$ and applies, at each layer $l$, an attention sublayer followed by a feed-forward network, both with residual connections:
\[
\begin{multlined}[c][0.98\textwidth]
h_i^{(l+0.5)} = W_O^{(l)} \cdot \mathrm{concat}_{h}\Big(\sum_{j} s^{(l)}_{h,i,j}\cdot W_v^{(l,h)}\, h_j^{(l)}\Big), \qquad \tilde{h}_i^{(l)} = h_i^{(l+0.5)} + b^{(l)} + h_i^{(l)}, \\
h_i^{(l+1)} = \mathrm{FF}^{(l)}\big(\tilde{h}_i^{(l)}\big) + \tilde{h}_i^{(l)},
\end{multlined}
\]
where the inner sums range over the context window and the concatenation runs over the $H$ heads. The learned parameters comprise the embedding and output head, query and key matrices $W_q^{(l,h)},W_k^{(l,h)}\in\mathbb{R}^{d_k\times d}$, value matrices $W_v^{(l,h)}\in\mathbb{R}^{d_v\times d}$, output matrices $W_O^{(l)}\in\mathbb{R}^{d\times H d_v}$, biases $b^{(l)}$, and feed-forward networks $\mathrm{FF}^{(l)}$. A relative positional function $\mathrm{pos}^{(l)}:\mathbb{N}\to\mathbb{R}^{d_k}$ is specified separately and enters the \emph{key side} of the scores: the weights $s^{(l)}_{h,i,\cdot}$ form the hardmax distribution (all weight on the score's maximizer) over scores $\langle W_q^{(l,h)}h_i^{(l)},\,W_k^{(l,h)}h_j^{(l)}+\mathrm{pos}^{(l)}(i-j)\rangle$. The next token is $v_{i+1}=\mathrm{argmax}(W_{\mathrm{out}}h_i^{(L)})$ at the last layer, for a decoder of depth $L$.

In the \textbf{sliding-window} setting with window $w$, position $i$ attends to $\{i-w+1, \ldots, i\}\cap \mathbb{N}$. The \textbf{KV cache} is indexed by (position, layer): $W_k^{(l,h)} h_j^{(l)}$ and $W_v^{(l,h)} h_j^{(l)}$ are written and fixed when position $j$ completes layer $l$.

\begin{remark}\label{rem:span}
With sliding-window attention and a persistent KV cache, a layer-$L$ transformer can depend on $L(w-1)+1$ positions, more than the visible window. This span alone does not provide the alignment needed for simulation: each queue value must still sit in the correct iteration's cache (Section~\ref{sec:example}).
\end{remark}

\subsection{Looped transformers}

Looped transformers re-apply a fixed unit instead of stacking independently parameterized layers \citep{giannou2023looped, yang2024looped, geiping2025scaling}. We adopt a \emph{weight-tying} version \citep{dehghani2019universal, lan2020albert}: the unit acts $L$ times on the current position's residual stream per generated token, with a separate per-iteration KV cache carrying state across positions.

\begin{definition}[looped transformer]\label{def:looped}
A \emph{looped transformer} is a decoder-only transformer consisting of

\begin{itemize}
\item a \textbf{loop unit} $u_{\varphi}$: a stack of $p \ge 1$ decoder layers specified by the shared weights $\varphi$;
\item the embedding $\mathrm{emb}$ and the output head $W_{\mathrm{out}}$, applied outside the loop;
\end{itemize}

together with a \textbf{loop count} $L \in \mathbb{N}$. Given $L$, the forward pass computes
\[
\begin{aligned}
h_i^{(0)} &= \mathrm{emb}(v_i), \qquad
h_i^{(l+1)} = u_{\varphi}\big(h_i^{(l)};\, \mathcal{C}_i^{(l)}\big)
\quad (0\le l<L),\\
v_{i+1} &= \mathrm{argmax}\big(W_{\mathrm{out}}\, h_i^{(L)}\big),
\end{aligned}
\]
the $l$-th application of $u_\varphi$ being iteration $l$ (layer and iteration are used interchangeably). Unrolling the loop yields a $pL$-layer transformer in which each iteration applies the same $p$ layers. Shared parameters are independent of $L$.
\end{definition}

\paragraph{Reading rule (per-iteration cache).} The cache $\mathcal{C}_i^{(l)}$ holds the iteration-$l$ key and value entries of the retained positions $j<i$. Iteration $l$ of position $i$ attends \emph{only} to those same-iteration entries; the current position's key and value are computed from $h_i^{(l)}$ and are frozen when written. A fixed iteration never reads another iteration's cache. Section~\ref{sec:example} shows the discipline on one queue, and the construction has to follow it.

A looped transformer is \emph{constant-bit-size} if the number of distinct learned parameters, their bit precision, and the bit precision of scalar intermediate values are all $O(1)$, independently of input length, loop count, and generation length.

\section{The hardmax construction}\label{sec:hardmax}

We first construct the simulation with hardmax attention; Section~\ref{sec:softmax} realizes the same selection exactly under a finite-precision softmax kernel. Throughout this section, queue $r$ has length $\ell_r\equiv1\pmod L$ (queue lengths satisfying this congruence are constructed in Appendix~\ref{app:congruence}). Set $w=\max_r w_r$ and $s_{\max}=\max_r\ell_r$. Each forward pass executes $L$ consecutive machine steps. Head $r$ reads the front of queue $r$ at distance $w_r-1$: attention retrieves that front by a fixed relative offset, and the feed-forward network applies the transition.

\subsection{Tokens and residual layout}\label{sec:tokens}

At nonterminal position $i$, the token is $(q_{Li},\sigma_{Li})$, where $\sigma_{Li}=(\sigma^{(0)}_{Li},\dots,\sigma^{(K-1)}_{Li})\in\Sigma^K$ is the tuple appended at log index $Li$. Generated tokens are thus machine-state symbols, not natural-language steps.

The embedding uses the layout specified in Appendix~\ref{app:components}:
\[
h_i^{(0)}
=
\bigl(1,\ e_{q_{Li}},\ \mathrm{onehot}(\sigma^{(0)}_{Li}),\dots,\mathrm{onehot}(\sigma^{(K-1)}_{Li}),\ 0^{|\Sigma|K}\bigr).
\]
The blocks store respectively a constant, the current state, the appended tuple, and $K$ scratch registers for retrieved fronts, giving $d=2|\Sigma|K+|Q|+1=\Theta(K)$. At iteration $l$, attention writes the fronts and the feed-forward network applies $\delta$, leaving the residual equal to the embedding of $(q_{Li+l+1},\sigma_{Li+l+1})$. After iteration $L-1$, the output head emits $(q_{L(i+1)},\sigma_{L(i+1)})$, or ACCEPT/REJECT if the state is halting, and the generation stops.

\subsection{Positional encoding and the attention heads}\label{sec:positional}

The relative positional encoding is the simplified form of \citet[Eq.~(2)]{liwang2025efficient}: writing $e_r$ for the basis vector of $\mathbb{R}^{d_k}$ indexed by head number $r$ (Appendix~\ref{app:components} fixes the indexing conventions), $d_k=K$,
\[
\mathrm{pos}(i-j)\;=\;\Delta\!\!\sum_{r:\, i-j=w_r-1}\!\! e_r ,
\]
whose $r$-th coordinate is $\Delta$ exactly at distance $i-j=w_r-1$. The query of head $r$ is the constant $e_r$ and the key is set to be $0$. The positional encoding enters the key side of the scores (Section~\ref{sec:model}). Thus
\[
\mathrm{score}_r(i,j)
=
\left\langle e_r,\ \mathrm{pos}(i-j)\right\rangle
=
\Delta[\,i-j=w_r-1\,].
\]
For every pass $i\ge w-1$, $j^*=i-w_r+1$ is the unique maximizer, so hardmax selects queue $r$'s designated position: the offset $w_r-1$ encodes ``front of queue $r$ at the current iteration,'' and Lemma~\ref{lem:invariant} proves that the entry selected is the front consumed at step $Li+l$. Equal-length queues remain separate via distinct heads and value projections. Appendix~\ref{app:components} gives the matrices.

\paragraph{Value entries and transition.} Let $V_r(j,l)=W_v^{(r)} h_j^{(l)}\in\mathbb{R}^{d_v}$, $d_v=|\Sigma|$, be head $r$'s cached value at position $j$ and iteration $l$, and let $V(j,l)$ denote the represented $K$-tuple. Each $W_v^{(r)}$ projects symbol block $r$; $W_O$ writes the selected values into the front blocks, and the feed-forward network applies $\delta$. All $L$ iterations share the same positional encoding, value projections, and $\delta$ lookup, so the construction satisfies Definition~\ref{def:looped}.

\subsection{The loop invariant}\label{sec:initialization}

\paragraph{Initialization.} Generation begins at pass $w-1$ with a loaded cache: for $0\le i\le w-2$ and $0\le l\le L-1$, set the cached value entry $V(i,l)=\sigma_{Li+l}$, and set the current token to $(q_{L(w-1)},\sigma_{L(w-1)})$. Since $L(w-1)=s_{\max}-1$, the first simulated transition occurs at step $s_{\max}-1$; the remaining run length is $\tau(n)$ (Definition~\ref{def:log}).

We treat this initialization as an explicit \emph{interface assumption}. It prescribes which entries a working system must place in the cache before generation; the cost of producing them is not included in any bound of this paper, and we do not prove that a natural-order prompt prefill computes them.

\begin{lemma}\label{lem:invariant}
For every position $i$ and iteration $l$, the value entry $V(i, l)$ encodes the $K$-tuple at log index $Li + l$: block $r$ holds $\sigma^{(r)}_{Li+l}$.
\end{lemma}

\begin{proof}
The initialization gives the claim for $0\le i\le w-2$ and all $l$, and the token at position $w-1$ supplies its iteration-$0$ entry $V(w-1,0)=\sigma_{L(w-1)}$.

Consider pass $i\ge w-1$ and iteration $l$, and suppose the residual entering iteration $l$ contains $(q_{Li+l},\sigma_{Li+l})$. Head $r$ attends to position $i-w_r+1$, whose iteration-$l$ entry was written in an earlier pass and is therefore the log entry
\[
V(i-w_r{+}1,\, l)^{(r)} \;=\; \sigma^{(r)}_{L(i-w_r+1)+l} \;=\; \sigma^{(r)}_{Li + l - \ell_r + 1},
\]
the front of queue $r$ for step $Li + l$, because $\ell_r - 1 = L(w_r - 1)$. The residual stream supplies $q_{Li+l}$, and the feed-forward network computes
\[
(q_{Li+l+1},\ \sigma_{Li+l+1}) \;=\; \delta\big(q_{Li+l},\ \sigma^{(0)}_{Li+l-\ell_0+1}, \dots, \sigma^{(K-1)}_{Li+l-\ell_{K-1}+1}\big).
\]
For $l<L-1$, this result is cached as $V(i,l+1)$, and the residual entering iteration $l+1$ contains $(q_{Li+l+1},\sigma_{Li+l+1})$; for $l=L-1$, the output head emits $(q_{L(i+1)},\sigma_{L(i+1)})$, whose embedding supplies $V(i+1,0)$. Every value entry is therefore the corresponding log entry, as claimed.
\end{proof}

Figure~\ref{fig:pass} illustrates one pass.

\begin{figure}[htbp]
\centering
\begin{tikzpicture}[
  font=\footnotesize,
  cbox/.style={draw, align=left, inner sep=2.5pt, text width=56mm},
  fbox/.style={draw, align=left, inner sep=2.5pt, text width=42mm},
  fdot/.style={draw, inner sep=2.5pt},
  arr/.style={-{Stealth}},
  darr/.style={-{Stealth}, dashed},
  gaplbl/.style={font=\scriptsize, align=center, inner sep=1.5pt}
]
\node[cbox] (TOK) {token at position $i$: $(q_{Li},\sigma_{Li})$, embedded as $V(i,0)=\sigma_{Li}$};
\node[cbox, below=2.5mm of TOK] (A0) {iteration $0$: $K$ heads read; head $r$ max at $j=i{-}w_r{+}1$};
\node[cbox, below=2.5mm of A0] (M0) {FFN: $\delta(q_{Li},\ \mathrm{fronts}\ \sigma^{(r)}_{Li-\ell_r+1})$ gives $(q_{Li+1},\sigma_{Li+1})$};
\node[cbox, below=2.5mm of M0] (DOTS) {$\dots$ one machine step per iteration: iteration $l$ writes $V(i,l{+}1)=\sigma_{Li+l+1}$ $\dots$};
\node[cbox, below=2.5mm of DOTS] (AL) {iteration $L{-}1$ on $V(i,L{-}1)=\sigma_{Li+L-1}$: head $r$ max at $j=i{-}w_r{+}1$};
\node[cbox, below=2.5mm of AL] (ML) {$\mathrm{FFN}$ gives $(q_{L(i+1)},\sigma_{L(i+1)})$};
\draw[arr] (TOK) -- (A0) -- (M0) -- (DOTS) -- (AL) -- (ML);
\node[cbox, below=4mm of ML] (EXIT) {output head emits token at position $i{+}1$: $(q_{L(i+1)},\sigma_{L(i+1)})$};
\draw[arr] (ML) -- (EXIT);
\node[fbox, anchor=east] (F0) at ($(A0.west)+(-33mm,0)$) {$V(i{-}w_r{+}1,0)$ block $r=\sigma^{(r)}_{L(i-w_r+1)}=\sigma^{(r)}_{Li-\ell_r+1}$};
\node[fbox, anchor=east] (FL) at ($(AL.west)+(-33mm,0)$) {$V(i{-}w_r{+}1,L{-}1)$ block $r=\sigma^{(r)}_{Li-\ell_r+L}$};
\node[fdot, anchor=center] (FDOTS) at ($(F0.south)!0.5!(FL.north)$) {$\dots$};
\draw (F0) -- (FDOTS) -- (FL);
\draw[darr] (F0.east) -- node[below, gaplbl, text width=40mm] {gap $\Delta$: head $r$ unique maximum at distance $w_r-1$} (A0.west);
\draw[darr] (FL.east) -- node[above, gaplbl] {gap $\Delta$} (AL.west);
\node[anchor=south west, font=\footnotesize\itshape, inner sep=1pt] (CURLBL) at ($(TOK.north west)+(-1.5mm,1mm)$) {position $i$: residual stream during pass $i$};
\draw[rounded corners=3pt] ($(CURLBL.north west)+(-1mm,1mm)$) rectangle ($(ML.south east)+(2mm,-1.5mm)$);
\node[anchor=south west, font=\footnotesize\itshape, inner sep=1pt, text width=44mm, align=left] (FRONTLBL) at ($(F0.north west)+(-1.5mm,1mm)$) {position $i{-}w_r{+}1$ (head $r$): frozen $V$ entries (written during pass $i{-}w_r{+}1$)};
\draw[rounded corners=3pt] ($(FRONTLBL.north west)+(-1mm,1mm)$) rectangle ($(FL.south east)+(2mm,-1.5mm)$);
\end{tikzpicture}
\caption{One pass at position $i$. At iteration $l$, head $r$ retrieves queue $r$'s front from $V(i-w_r+1,l)$, and the feed-forward network advances the machine by one step. The final write is emitted as the token at position $i+1$.}
\label{fig:pass}
\end{figure}
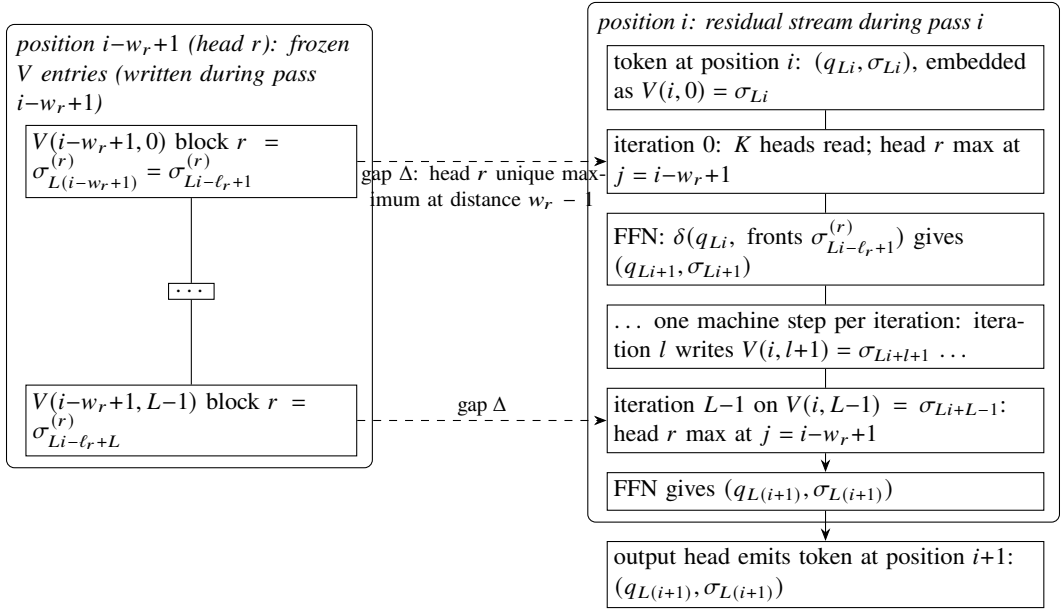

\subsection{The main theorem}\label{sec:main-theorem}

\begin{proposition}[queue-level construction]\label{prop:queue}
Let $M=(\Sigma,Q,\delta)$ be a synchronous $K$-queue machine over a fixed alphabet $\Sigma$, with absorbing halting states and lengths $\ell_r=L(w_r-1)+1$. Under the initialization of Section~\ref{sec:initialization}, the construction of Sections~\ref{sec:tokens}--\ref{sec:initialization} gives a decoder-only looped transformer with one shared block, $K$ heads, window $w=\max_r w_r$, embedding dimension $2|\Sigma|K+|Q|+1$, and constant bit-size. It emits the execution log at stride $L$ and therefore uses
\[
T_{\mathrm{cot}}=\left\lceil\frac{\tau(n)}{L}\right\rceil
\]
generated tokens, including at most $L-1$ absorbing transitions in the final pass.
\end{proposition}

\begin{proof}
The decoder block, $K$ heads, window $w$, dimension $2|\Sigma|K+|Q|+1$, and constant bit-size are those of Sections~\ref{sec:tokens}--\ref{sec:positional}. By Lemma~\ref{lem:invariant}, at every pass $i\ge w-1$ and every iteration $l$ the cached values are the log entries $\sigma_{Li+l}$, so head $r$ retrieves the front of queue $r$ at machine step $Li+l$ and the feed-forward network applies $\delta$. After $L$ iterations the residual is therefore the embedding of $(q_{L(i+1)},\sigma_{L(i+1)})$. Appendix~\ref{app:components} (\ref{appsub:output}) emits this pair as the token at position $i+1$, or ACCEPT/REJECT if the state is halting. Generation thus starts from the initialization of Section~\ref{sec:initialization} at index $L(w-1)=s_{\max}-1$ and follows the log at stride $L$ for the remaining $\tau(n)$ steps, giving $T_{\mathrm{cot}}=\lceil\tau(n)/L\rceil$ with at most $L-1$ absorbing transitions in the final pass.
\end{proof}

\begin{theorem}\label{thm:main}
Any $(t(n),s(n))$ time-space bounded $k'$-tape Turing machine can be simulated by a constant-bit-size looped transformer (with its input encoded in the initial KV cache as specified in Section~\ref{sec:initialization}) with head--layer product $K=8kk'$ and context window length $O(s(n)/L)$, that takes $O(s(n)^{1/k}/L)$ CoT steps for each simulated TM step, leading to a total CoT length of $O(t(n)\cdot s(n)^{1/k}/L)$, for every loop count $L=L(n)$ and every fixed integer $k\ge1$ with $s(n)^{1/k}=\Omega(L)$.
\end{theorem}

\begin{proof}
For the prescribed $L$ and the given fixed $k$ satisfying $s(n)^{1/k}=\Omega(L)$, the congruent lengths supplied by Appendix~\ref{app:congruence} (Lemma~\ref{lem:congruence}) give a step-equivalent $K=8kk'$-queue simulation with run length $\tau(n)=O(t(n)s(n)^{1/k})$ and window $w=O(s(n)/L)$. Proposition~\ref{prop:queue} yields
\[
T_{\mathrm{cot}}=O\!\left(\frac{t(n)s(n)^{1/k}}{L}\right).
\]
Multiplying by $L$ gives the serial-work bound. The cache stores $w$ positions at each of $L$ iterations, so total KV state is $\Theta(Lw)=\Theta(s(n))$.
\end{proof}

\paragraph{Resource accounting.} The factor-$L$ reduction applies to visible resources only: generated tokens and the visible window shrink by $L$, while the total cached KV state $\Theta(Lw)=\Theta(s(n))$ and the serial block work $L\cdot T_{\mathrm{cot}}=O(t(n)s(n)^{1/k})$ are unchanged from $L=1$. Table~\ref{tab:resources} tabulates both regimes. The $K$ heads of the single shared block are re-applied at each iteration, so the head--layer product does not grow with $L$.

\begin{table}[htbp]
\centering
\begin{tabular}{@{}lll@{}}
\toprule
Quantity & $L=1$ (endpoint) & admissible $L=O(s(n)^{1/k})$ \\
\midrule
Generated tokens & $O(t(n)s(n)^{1/k})$ & $O\!\big(t(n)s(n)^{1/k}/L\big)$, exactly $\lceil\tau(n)/L\rceil$ \\
Visible positions (window) & $O(s(n))$ & $O\!\big(s(n)/L\big)$ \\
Loop executions (serial block work) & $O(t(n)s(n)^{1/k})$ & $O(t(n)s(n)^{1/k})$, unchanged \\
Total cached KV state & $\Theta(s(n))$ & $\Theta(Lw)=\Theta(s(n))$, unchanged \\
Learned model size & $O(1)$, machine-dependent & $O(1)$, machine-dependent \\
\bottomrule
\end{tabular}
\caption{The construction at the endpoint $L=1$ (which recovers the bounds of \citealp{liwang2025efficient}) and at a general admissible loop count. The simulated Turing machine is the same, while the queue lengths are re-instantiated for each $L$; only the first two rows shrink with $L$.}
\label{tab:resources}
\end{table}

A concrete instance: the balanced point $s(n)=m^2$, $k=2$, $L=\Theta(m)$ is admissible because $s(n)^{1/k}=m=\Omega(L)$, and it gives $O(t(n))$ generated tokens, a window of $O(m)$ positions, $\Theta(m^2)$ cached entries, and $O(t(n)m)$ loop executions.

\subsection{A worked example}\label{sec:example}

Take one queue of length $\ell=7$ and loop count $L=3$. The window is then $w=3$: at pass $i$ and iteration $l$, the head reads position $i-2$, at relative distance $w-1=2$, in the same iteration. Since $\ell-1=6=L(w-1)$, the entry at position $i-2$ and iteration $l$ is the front consumed at step $3i+l$. The construction keeps a separate copy of this bookkeeping for every queue. One queue is enough to see the alignment.

Before generation, the loaded cache of Section~\ref{sec:initialization} must already hold $V(i,l)=\sigma_{Li+l}$ for $0\le i\le w-2=1$ and $l\in\{0,1,2\}$: $\sigma_0,\sigma_1,\sigma_2$ at position $0$ and $\sigma_3,\sigma_4,\sigma_5$ at position $1$, together with the token $(q_6,\sigma_6)$ at position $2$. Here $L(w-1)=6=s_{\max}-1$, and $s_{\max}=7$ is the length of this queue, so generation begins at step $6$, the first step whose front is already in the cache.

Table~\ref{tab:passes} lists the first passes. At pass $i=3$ the three iterations execute steps $9$, $10$, and $11$, and iteration $l$ reads $V(1,l)=\sigma_{3+l}$. On a queue of length $7$, the symbol consumed at step $m$ was appended at step $m-7+1$, so those fronts are the symbols appended at steps $3$, $4$, and $5$. In general, the symbol read at pass $i$ and iteration $l$ has log index
\[
3(i-2)+l \;=\; 3i+l-7+1 .
\]
Each pass advances the machine from log index $3i$ to $3i+3$ and emits the state at the boundary. The visible tokens $(q_6,\sigma_6),(q_9,\sigma_9),(q_{12},\sigma_{12}),\dots$ are the log sampled at stride $L=3$. Intermediate steps, such as steps $7$ and $8$ inside pass $2$, are never emitted.

\begin{table}[htbp]
\centering
\begin{tabular}{@{}cccccc@{}}
\toprule
Pass $i$ & Window & Steps executed & Reads ($l=0,1,2$) & Residual advances from & Emits \\
\midrule
2 & $\{0,1,2\}$ & 6, 7, 8 & $\sigma_0,\sigma_1,\sigma_2$ (at pos.\ 0) & 6, 7, 8 & $(q_9,\sigma_9)$ \\
3 & $\{1,2,3\}$ & 9, 10, 11 & $\sigma_3,\sigma_4,\sigma_5$ (at pos.\ 1) & 9, 10, 11 & $(q_{12},\sigma_{12})$ \\
4 & $\{2,3,4\}$ & 12, 13, 14 & $\sigma_6,\sigma_7,\sigma_8$ (at pos.\ 2) & 12, 13, 14 & $(q_{15},\sigma_{15})$ \\
\bottomrule
\end{tabular}
\caption{One queue with $L=3$, $\ell=7$, and $w=3$. At pass $i$ and iteration $l$, the head reads position $i-2$ in iteration $l$, which is log index $3(i-2)+l=3i+l-7+1$. The pass advances the machine by three steps and emits every third log state.}
\label{tab:passes}
\end{table}

Three features of the example are what the construction uses. First, iteration $1$ of pass $3$ reads $V(1,1)=\sigma_4$: the same iteration of an earlier position. It does not read an iteration-$0$ entry, and it does not mix iterations. Section~\ref{sec:model} states this per-iteration rule. Second, a value entry is frozen once written. Initialization supplies $V(1,\cdot)$; pass $2$ writes $V(2,\cdot)$, and pass $4$ reads that same entry. Third, the window has to be the one the congruence gives. The front required at step $11$ is $\sigma_5$, written at position $1$, exactly $w-1=2$ positions back. With a window of $2$, position $1$ would already lie outside the positions pass $3$ can read. We call a representation \emph{layer-aligned} when, at every simulated step, the inputs that step consumes already sit in the KV cache at the iteration, or layer, that executes the step.

\section{Softmax simulates hardmax at constant bit-size}\label{sec:softmax}

Section~\ref{sec:hardmax} uses hardmax attention. Theorem~\ref{thm:main} also holds for a max-first binary32 softmax kernel, with no growth in precision as the window or the loop count grows. The five-clause kernel contract, the underflow arithmetic, and the complete proof are in Appendix~\ref{app:softmax}.

Write $a_{i,j}$ for the attention score of candidate position $j$ in the row of query position $i$. For a row with unique maximizer $j^*$, define the attention score gap
\[
\Delta \;=\; a_{i,j^*} \;-\; \max_{j \ne j^*} a_{i,j},
\]
the score separation of \citep[Def.~D.1]{brosamle2026expressive}.

\paragraph{The kernel, in words.} The kernel evaluates the exponential only on inputs shifted by the row maximum, using a fixed binary32 exponential that returns exactly $1$ at $0$ and $+0$ at every input at most $-256$; numerator and denominator are accumulated in binary32 and divided once; the quotient is cast once at the output boundary, to a fixed format that contains $0$ and $1$ exactly. The role of the threshold is arithmetic: since $256 > 150 \ln 2$, we have $e^{-256} < 2^{-150}$, the round-to-nearest midpoint between $0$ and the least positive binary32 subnormal, so a correctly rounded binary32 exponential returns $+0$ on every input at most $-256$ (Appendix~\ref{app:softmax}).

\begin{lemma}[softmax implements unique hardmax exactly]\label{lem:softmax}
Let $(a_j)_j$ be a score row with unique maximizer $j^*$ and score gap at least $256$. Then the max-first kernel's computed attention output is exactly $v_{j^*}$, independently of the row length, provided $v_{j^*}$ is exactly representable in the accumulation format and preserved by the output cast---in particular for the values in $\{0,1\}$ used in this construction.
\end{lemma}

\begin{proof}
By the kernel, the shifted inputs are $\hat a_{j^*} = 0$ and $\hat a_j \le -256$ for every $j \ne j^*$. The exponential row is exactly the one-hot indicator of $j^*$: $E_{32}(\hat a_{j^*}) = 1$ and $E_{32}(\hat a_j) = +0$ for every $j \ne j^*$. Numerator and denominator are therefore the exact sums $v_{j^*}$ and $1$, each partial sum containing only zeros and one nonzero term. Division and the single boundary cast preserve the result.
\end{proof}

\begin{corollary}[the main results hold under softmax]\label{cor:softmax}
Set $\Delta = 256$ in Section~\ref{sec:positional}. Then every score row has scores in $\{256, 0\}$, a unique maximizer at distance $w_r - 1$, and shifted losers exactly $-256$, all exactly representable in binary32. The value vectors are one-hot, and the remaining parameters and feed-forward computations use a fixed machine-dependent format as specified in Appendix~\ref{app:components}. By Lemma~\ref{lem:softmax} every attention output equals the hardmax one-hot of Section~\ref{sec:positional} bit for bit, so Theorem~\ref{thm:main} and Proposition~\ref{prop:queue} hold with softmax attention under the max-first kernel: the parameters are those of the hardmax construction rescaled by the fixed ratio between its score gap and the kernel constant $256$ (here $1$, since $\Delta=256$), the resources are unchanged, and the bit-size is still $O(1)$, independent of $n$ and $L$.
\end{corollary}

\begin{remark}[scope]\label{rem:softmax-scope}
The result is conditional on the stated kernel, in three ways. First, it requires a binary32 exponential that returns $+0$ below a fixed threshold, without a positive floor, logit soft-capping, or an additive $\varepsilon$ in the normalizer; the exact contract is Appendix~\ref{app:softmax}. Second, the arithmetic is fixed-format: values not exactly representable in it (arbitrary real value vectors, for instance) are outside the exactness statement. Third, no claim is made about kernels beyond the contract: subtract-max is standard in production softmax implementations (Appendix~\ref{app:softmax} records examples), but that observation is context, not a conformance proof, and in particular online and fused attention kernels with their own epilogues are not thereby covered. Appendix~\ref{app:softmax} also treats two variants: a kernel that does not subtract the maximum (then Section~\ref{sec:positional} can set the winner's score to $0$ and every other score to $-256$), and a kernel with a positive scale, canonically $1/\sqrt{d_k}$ (then instantiate $\Delta=\lceil 256\sqrt{d_k}\rceil$; here $d_k=K$ is a machine constant, so $\Delta$ remains $O(1)$). Nothing generalizes to exact-real softmax.
\end{remark}

The comparison with exact-real arithmetic is a comparison of assumptions. \citet{brosamle2026expressive} simulate Turing machines with low-precision softmax and summarized CoT, separating the winner from the other scores with depth, width, and precision that grow logarithmically. The present construction fixes a constant score gap and lets finite-format underflow make the selection exact; the only added cost is that constant gap $\Delta$. These are different precision and architecture assumptions. We do not claim that their arithmetic implements exact hard selection.

\section{Related work}\label{sec:related}

We keep this section argument-driven; Appendix~\ref{app:comparison} carries the detailed comparison table of transformers simulating $(t,s)$-bounded Turing machines, with per-row notes and normalized notation.

\paragraph{Transformer expressivity.} RASP models transformers as readable programs \citep{weiss2021thinking, lindner2023tracr}. Variants such as B-RASP, C-RASP, and RASP-L, together with subsequent formal-language analyses, characterize the languages recognized under different architectural restrictions \citep{angluin2024masked, yang2024counting, zhou2024algorithms, strobl2024formal, strobl2025transducers, huang2025formal}; see \citet{strobl2024formal} for a survey. On the circuit side, \citet{merrill2023parallelism} show that constant-depth, log-precision transformers without CoT are limited to $\mathsf{TC}^0$-style shallow circuits. \citet{feng2023towards} connect CoT length to circuit depth: polynomially long CoT lifts transformers from shallow-circuit problems to all of $\mathsf{P}$.

\paragraph{Looped transformers and latent thought.} Looped transformers extend serial computation by re-applying a shared block \citep{giannou2023looped, yang2024looped, geiping2025scaling}; \citet{merrill2025little} show that logarithmically deep looped transformers recognize regular languages and solve graph connectivity without CoT. \citet{saunshi2025reasoning} study how latent thought can simulate computations expressed through CoT. \citet{xusato2026formal} establish simulation results and complexity-class characterizations for CoT and latent thought, together with conditional separations: latent iterations support more efficient parallel computation in their deterministic setting, while stochastic CoT has advantages for approximate counting and sampling under their stated assumptions. Their deterministic comparison allows the embedding dimension to grow with input length. Our result is complementary: for a fixed Turing machine and hierarchy parameter, we construct a family with fixed machine-dependent learned size and precision in which varying $L$ exchanges explicit CoT steps for latent iterations. We account for generated tokens, visible context, serial loop work, and persistent per-iteration KV state under the loaded-cache interface. This is an interpolation for the constructed recurrent simulation, not an equivalence of arbitrary CoT and latent-thought models or a claim that all intermediate tokens can be eliminated.

\paragraph{Autoregressive Turing-machine simulations.} Autoregressive transformers are Turing-complete \citep{perez2021attention, bhattamishra2020computational, merrill2024expressive, li2024chain, qiu2025ask}, with precisions and dimensions ranging from unbounded to constant. \citet{jiang2025softmax} establish softmax Turing completeness via counter machines with full-history attention and no fixed numerical-precision bound. Two later lines of work bring the context down to the space bound. Context-management schemes erase or summarize intermediate CoT: PENCIL \citep{yang2025pencil} achieves maximal context $O(s(n))$, and \citet{cui2026position} show that the management policy itself determines the computable class; \citet{brosamle2026expressive} keep every segment at $O(s(n))$ with low-precision softmax, at logarithmically growing depth, width, and precision. Sliding-window attention uses locality in place of erasure: \citet{schuurmans2024autoregressive} gave the earliest bounded-window universality, with a circular-queue context; \citet{liwang2025constant} prove constant-bit-size Turing completeness for an $O(s(n))$ window via a Post-machine simulation and characterize $\mathsf{WINDOW}[s(n)] = \mathsf{SPACE}[s(n)]$; \citet{limao2026rethinking} remove the need for positional encodings; and \citet{liwang2025efficient} reach $O(s(n)^{1/k})$ CoT tokens per TM step, with window $s(n)+1$, by increasing the head--layer product. Our construction starts from that last simulation and folds $L$ of its transitions into each token by making the queue representation layer-aligned (Section~\ref{sec:example}). At $L=1$ it recovers that simulation's bounds.

\section{Discussion}\label{sec:discussion}

The following limitations bound the result.

\begin{itemize}
\item \emph{Loaded cache.} The input is pre-loaded into the initial KV cache (Section~\ref{sec:initialization}); the initialization cost is excluded from all bounds, and ordinary prompt prefill is not proved to produce the required entries.
\item \emph{Visible resources only.} The factor-$L$ reduction applies to generated tokens and to the visible window only; total cached KV state and serial block work are unchanged from $L=1$.
\item \emph{Kernel-conditional softmax.} The softmax result holds only under the max-first binary32 contract of Appendix~\ref{app:softmax}, for values exactly representable in the accumulation and output formats; it does not extend to exact-real softmax, to kernels with positive floors or logit soft-capping, or to production kernels beyond that contract.
\item \emph{Generated tokens are machine states.} The visible trace is a sequence of internal machine-state symbols at stride $L$, not a natural-language chain of thought; CoT-length statements concern this trace.
\item \emph{No learning claims.} Nothing here shows that such computations are learned by pretraining, reinforcement learning, or any other training signal; the messages of Section~\ref{sec:intro} are conjectural, and the supervision question for strided traces is open.
\item \emph{Bounds, not measurements.} All quantitative statements are asymptotic bounds of the construction; no experiment, speedup, or deployment claim is made.
\end{itemize}

\section*{AI use statement}
We used generative AI tools to improve the readability of the manuscript
and to identify relevant literature. We also used the Paper Assistant
Tool (PAT) to review the paper's technical content. After assessing its
feedback, we corrected errors it identified and cited relevant papers
it suggested. The authors reviewed the AI-assisted revisions and take
responsibility for the final text, citations, and mathematical claims.

\bibliographystyle{plainnat}
\bibliography{refs}

@article{perez2021attention,
  author  = {P{\'e}rez, Jorge and Barcel{\'o}, Pablo and Marinkovic, Javier},
  title   = {Attention is {Turing}-complete},
  journal = {Journal of Machine Learning Research},
  volume  = {22},
  number  = {75},
  pages   = {1--35},
  year    = {2021}
}

@inproceedings{bhattamishra2020computational,
  author    = {Bhattamishra, Satwik and Patel, Arkil and Goyal, Navin},
  title     = {On the computational power of transformers and its implications in sequence modeling},
  booktitle = {Proceedings of the 24th Conference on Computational Natural Language Learning (CoNLL)},
  pages     = {455--475},
  year      = {2020}
}

@article{merrill2023parallelism,
  author  = {Merrill, William and Sabharwal, Ashish},
  title   = {The parallelism tradeoff: Limitations of log-precision transformers},
  journal = {Transactions of the Association for Computational Linguistics},
  volume  = {11},
  pages   = {531--545},
  year    = {2023}
}

@inproceedings{feng2023towards,
  author    = {Feng, Guhao and Zhang, Bohang and Gu, Yuntian and Ye, Haotian and He, Da-Wei and Wang, Liwei},
  title     = {Towards revealing the mystery behind chain of thought: a theoretical perspective},
  booktitle = {Advances in Neural Information Processing Systems 36 (NeurIPS 2023)},
  pages     = {70757--70798},
  year      = {2023}
}

@article{post1936finite,
  author  = {Post, Emil L.},
  title   = {Finite combinatory processes --- formulation 1},
  journal = {Journal of Symbolic Logic},
  volume  = {1},
  number  = {3},
  pages   = {103--105},
  year    = {1936}
}

@article{strobl2024formal,
  author  = {Strobl, Lena and Merrill, William and Weiss, Gail and Chiang, David and Angluin, Dana},
  title   = {What formal languages can transformers express? {A} survey},
  journal = {Transactions of the Association for Computational Linguistics},
  volume  = {12},
  pages   = {543--561},
  year    = {2024}
}

@inproceedings{shaw2018self,
  author    = {Shaw, Peter and Uszkoreit, Jakob and Vaswani, Ashish},
  title     = {Self-attention with relative position representations},
  booktitle = {Proceedings of the 2018 Conference of the North American Chapter of the Association for Computational Linguistics: Human Language Technologies, Volume 2 (Short Papers)},
  pages     = {464--468},
  year      = {2018}
}

@inproceedings{giannou2023looped,
  author    = {Giannou, Angeliki and Rajput, Shashank and Sohn, Jy-yong and Lee, Kangwook and Lee, Jason D. and Papailiopoulos, Dimitris},
  title     = {Looped transformers as programmable computers},
  booktitle = {Proceedings of the 40th International Conference on Machine Learning (ICML 2023)},
  series    = {Proceedings of Machine Learning Research},
  volume    = {202},
  pages     = {11398--11442},
  year      = {2023}
}

@inproceedings{weiss2021thinking,
  author    = {Weiss, Gail and Goldberg, Yoav and Yahav, Eran},
  title     = {Thinking like transformers},
  booktitle = {Proceedings of the 38th International Conference on Machine Learning (ICML 2021)},
  series    = {Proceedings of Machine Learning Research},
  volume    = {139},
  pages     = {9979--9991},
  year      = {2021}
}

@inproceedings{wei2022chain,
  author    = {Wei, Jason and Wang, Xuezhi and Schuurmans, Dale and Bosma, Maarten and Xia, Fei and Chi, Ed H. and Le, Quoc V. and Zhou, Denny},
  title     = {Chain-of-Thought Prompting Elicits Reasoning in Large Language Models},
  booktitle = {Advances in Neural Information Processing Systems 35 (NeurIPS 2022)},
  pages     = {24824--24837},
  year      = {2022}
}

@inproceedings{dehghani2019universal,
  author    = {Dehghani, Mostafa and Gouws, Stephan and Vinyals, Oriol and Uszkoreit, Jakob and Kaiser, {\L}ukasz},
  title     = {Universal transformers},
  booktitle = {Proceedings of the 7th International Conference on Learning Representations (ICLR)},
  year      = {2019}
}

@inproceedings{lan2020albert,
  author    = {Lan, Zhenzhong and Chen, Mingda and Goodman, Sebastian and Gimpel, Kevin and Sharma, Piyush and Soricut, Radu},
  title     = {{ALBERT}: {A} lite {BERT} for self-supervised learning of language representations},
  booktitle = {Proceedings of the 8th International Conference on Learning Representations (ICLR)},
  year      = {2020}
}

@inproceedings{vaswani2017attention,
  author    = {Vaswani, Ashish and Shazeer, Noam and Parmar, Niki and Uszkoreit, Jakob and Jones, Llion and Gomez, Aidan N. and Kaiser, {\L}ukasz and Polosukhin, Illia},
  title     = {Attention is all you need},
  booktitle = {Advances in Neural Information Processing Systems 30 (NIPS 2017)},
  year      = {2017}
}

@inproceedings{li2024chain,
  author    = {Li, Zhiyuan and Liu, Hong and Zhou, Denny and Ma, Tengyu},
  title     = {Chain of thought empowers transformers to solve inherently serial problems},
  booktitle = {Proceedings of the 12th International Conference on Learning Representations (ICLR)},
  year      = {2024},
  note      = {arXiv:2402.12875}
}

@inproceedings{merrill2024expressive,
  author    = {Merrill, William and Sabharwal, Ashish},
  title     = {The expressive power of transformers with chain of thought},
  booktitle = {Proceedings of the 12th International Conference on Learning Representations (ICLR)},
  year      = {2024},
  note      = {arXiv:2310.07923}
}

@inproceedings{merrill2025little,
  author    = {Merrill, William and Sabharwal, Ashish},
  title     = {A little depth goes a long way: The expressive power of log-depth transformers},
  booktitle = {Advances in Neural Information Processing Systems 38 (NeurIPS 2025)},
  year      = {2025},
  note      = {arXiv:2503.03961}
}

@inproceedings{saunshi2025reasoning,
  author    = {Saunshi, Nikunj and Dikkala, Nishanth and Li, Zhiyuan and Kumar, Sanjiv and Reddi, Sashank J.},
  title     = {Reasoning with latent thoughts: On the power of looped transformers},
  booktitle = {Proceedings of the 13th International Conference on Learning Representations (ICLR)},
  year      = {2025},
  note      = {arXiv:2502.17416}
}

@inproceedings{xusato2026formal,
  author    = {Xu, Kevin and Sato, Issei},
  title     = {A formal comparison between chain of thought and latent thought},
  booktitle = {Proceedings of the 43rd International Conference on Machine Learning (ICML)},
  year      = {2026},
  note      = {arXiv:2509.25239, camera-ready version}
}

@inproceedings{qiu2025ask,
  author    = {Qiu, Ruizhong and Xu, Zhe and Bao, Wenxuan and Tong, Hanghang},
  title     = {Ask, and it shall be given: On the {Turing} completeness of prompting},
  booktitle = {Proceedings of the 13th International Conference on Learning Representations (ICLR)},
  year      = {2025},
  note      = {arXiv:2411.01992}
}

@inproceedings{liwang2025constant,
  author    = {Li, Qian and Wang, Yuyi},
  title     = {Constant bit-size transformers are {Turing} complete},
  booktitle = {Advances in Neural Information Processing Systems 38 (NeurIPS 2025)},
  pages     = {69254--69273},
  year      = {2025},
  doi       = {10.52202/085713-2083}
}

@misc{liwang2025efficient,
  author = {Li, Qian and Wang, Yuyi},
  title  = {Efficient {Turing} machine simulation with transformers},
  year   = {2025},
  note   = {arXiv:2512.00003}
}

@misc{limao2026rethinking,
  author = {Li, Qian and Mao, Xinyu and Teng, Shang-Hua},
  title  = {Rethinking the role of positional encoding: Sliding-window transformers without {PE} remain {Turing} complete},
  year   = {2026},
  note   = {arXiv:2606.01532}
}

@misc{schuurmans2024autoregressive,
  author = {Schuurmans, Dale and Dai, Hanjun and Zanini, Franco},
  title  = {Autoregressive large language models are computationally universal},
  year   = {2024},
  note   = {arXiv:2410.03170}
}

@inproceedings{yang2025pencil,
  author    = {Yang, Chenxiao and Srebro, Nathan and McAllester, David and Li, Zhiyuan},
  title     = {{PENCIL}: Long thoughts with short memory},
  booktitle = {Proceedings of the 42nd International Conference on Machine Learning (ICML 2025)},
  series    = {Proceedings of Machine Learning Research},
  volume    = {267},
  pages     = {71092--71156},
  year      = {2025},
  publisher = {PMLR}
}

@misc{brosamle2026expressive,
  author = {Br{\"o}samle, Moritz and Eckstein, Stephan},
  title  = {The expressive power of low precision softmax transformers with (summarized) chain-of-thought},
  year   = {2026},
  note   = {arXiv:2605.18079}
}

@inproceedings{jiang2025softmax,
  author    = {Jiang, Hongjian and Hahn, Michael and Zetzsche, Georg and Lin, Anthony Widjaja},
  title     = {Softmax transformers are {Turing}-complete},
  booktitle = {Proceedings of the 14th International Conference on Learning Representations (ICLR)},
  year      = {2026},
  note      = {arXiv:2511.20038}
}

@misc{cui2026position,
  author = {Cui, Guanyu and Wei, Zhewei and He, Kun},
  title  = {Position: The {Turing}-completeness of real-world autoregressive transformers relies heavily on context management},
  year   = {2026},
  note   = {arXiv:2605.19514, accepted to the ICML 2026 Position Paper Track}
}

@inproceedings{yang2024looped,
  author    = {Yang, Liu and Lee, Kangwook and Nowak, Robert D. and Papailiopoulos, Dimitris},
  title     = {Looped transformers are better at learning learning algorithms},
  booktitle = {Proceedings of the 12th International Conference on Learning Representations (ICLR)},
  year      = {2024},
  note      = {arXiv:2311.12424}
}

@inproceedings{geiping2025scaling,
  author    = {Geiping, Jonas and McLeish, Sean and Jain, Neel and Kirchenbauer, John and Singh, Siddharth and Bartoldson, Brian R. and Kailkhura, Bhavya and Bhatele, Abhinav and Goldstein, Tom},
  title     = {Scaling up test-time compute with latent reasoning: A recurrent depth approach},
  booktitle = {Advances in Neural Information Processing Systems 38 (NeurIPS 2025)},
  pages     = {46242--46293},
  year      = {2025}
}

@inproceedings{lindner2023tracr,
  author    = {Lindner, David and Kram{\'a}r, J{\'a}nos and Rahtz, Matthew and McGrath, Thomas and Mikulik, Vladimir},
  title     = {Tracr: Compiled transformers as a laboratory for interpretability},
  booktitle = {Advances in Neural Information Processing Systems 36 (NeurIPS 2023)},
  year      = {2023},
  note      = {arXiv:2301.05062}
}

@inproceedings{angluin2024masked,
  author    = {Angluin, Dana and Chiang, David and Yang, Andy},
  title     = {Masked hard-attention transformers recognize exactly the star-free languages},
  booktitle = {Proceedings of the Conference on Language Modeling (COLM)},
  year      = {2024},
  note      = {arXiv:2310.13897; introduces Boolean RASP (B-RASP)}
}

@inproceedings{yang2024counting,
  author    = {Yang, Andy and Chiang, David},
  title     = {Counting like transformers: Compiling temporal counting logic into softmax transformers},
  booktitle = {Proceedings of the Conference on Language Modeling (COLM)},
  year      = {2024},
  note      = {arXiv:2404.04393; introduces C-RASP}
}

@inproceedings{zhou2024algorithms,
  author    = {Zhou, Hattie and Bradley, Arwen and Littwin, Etai and Razin, Noam and Saremi, Omid and Susskind, Josh and Bengio, Samy and Nakkiran, Preetum},
  title     = {What algorithms can transformers learn? A study in length generalization},
  booktitle = {Proceedings of the 12th International Conference on Learning Representations (ICLR)},
  year      = {2024},
  note      = {arXiv:2310.16028; introduces RASP-L}
}

@article{strobl2025transducers,
  author    = {Strobl, Lena and Angluin, Dana and Chiang, David and Rawski, Jonathan and Sabharwal, Ashish},
  title     = {Transformers as transducers},
  journal   = {Transactions of the Association for Computational Linguistics},
  volume    = {13},
  pages     = {200--219},
  year      = {2025},
  doi       = {10.1162/tacl_a_00736},
  note      = {arXiv:2404.02040; extends B-RASP to transductions}
}

@inproceedings{huang2025formal,
  author    = {Huang, Xinting and Yang, Andy and Bhattamishra, Satwik and Sarrof, Yash and Krebs, Andreas and Zhou, Hattie and Nakkiran, Preetum and Hahn, Michael},
  title     = {A formal framework for understanding length generalization in transformers},
  booktitle = {Proceedings of the 13th International Conference on Learning Representations (ICLR)},
  year      = {2025},
  note      = {arXiv:2410.02140}
}

\appendix

\section{Congruent queue lengths}\label{app:congruence}

Grouping $L$ consecutive steps of a synchronous $K$-queue machine into a single block requires that the front of queue $r$ at global step $Li+l$ be the symbol written at the same intra-block offset of an earlier block. That symbol was written at step $Li+l-\ell_r+1$, so the alignment
\[
Li+l-\ell_r+1 \;=\; L(i-w_r+1)+l
\]
holds for all $i$ and $0\le l<L$ if and only if
\[
\ell_r\equiv1\pmod L,
\qquad
w_r=\frac{\ell_r-1}{L}+1 .
\]

In the conversion of Fact~\ref{fact:queue}, level $g$ has two content-half queues of length $B^g$ and one buffer queue of length $2B^g$, where $B=\lceil s(n)^{1/k}\rceil$. We split the buffer into two queues of length $B^g$, so that all four physical queues per stack at the level have the same length.

\begin{lemma}\label{lem:congruence}
Prescribe $L=L(n)\ge1$ and choose a fixed integer $k\ge1$ satisfying $s(n)^{1/k}=\Omega(L)$. Let $B=\lceil s(n)^{1/k}\rceil$ and
\[
b=\min\bigl\{\,b'\ge\max\{B,\ L+1\}\ :\ b'\equiv1\pmod L\,\bigr\},
\qquad
\ell_g=b^{\,g}\qquad(1\le g\le k).
\]
Then the conversion can be implemented with $K=8kk'$ queues, all four queues per stack at level $g$ having length $\ell_g$. Every queue length satisfies $\ell_g\equiv1\pmod L$, and the resulting machine runs in $O(t(n)s(n)^{1/k})$ steps using $O(k's(n))$ queue space. Moreover,
\[
w_r=\frac{\ell_r-1}{L}+1\ge2,
\qquad
w=O\!\left(\frac{s(n)}{L}\right).
\]
\end{lemma}

\begin{proof}
Run the conversion of Fact~\ref{fact:queue} with geometric base $b$. Replacing $B$ by $b\ge B$ preserves the divisibility between consecutive levels and provides at least the original capacity. Split each length-$2\ell_g$ buffer, as any queue of length $m_1+m_2$ can be simulated step-for-step by two queues of lengths $m_1$ and $m_2$; each stack now uses four length-$\ell_g$ queues per level, so a $k'$-tape machine uses $2\cdot4kk'=8kk'$ queues. The split simulation is step-for-step equivalent to the original one.

By minimality,
\[
b\le\max\{B,L+1\}+L-1=O(s(n)^{1/k}),
\]
because $L=O(s(n)^{1/k})$. Together with $b\ge B$, this gives $b=\Theta(s(n)^{1/k})$ and $\ell_k=b^k=\Theta(s(n))$. The conversion therefore takes $O(t(n)b)=O(t(n)s(n)^{1/k})$ steps and $O(k'\sum_{g=1}^k\ell_g)=O(k's(n))$ queue space. Finally, $b\equiv1\pmod L$ implies $\ell_g=b^g\equiv1\pmod L$, while $b\ge L+1$ gives $w_r\ge2$; hence $w=\Theta(\ell_k/L)=O(s(n)/L)$.
\end{proof}

The construction introduces no additional transitions, so $\tau(n)=O(t(n)s(n)^{1/k})$, the run length used in Proposition~\ref{prop:queue} and Theorem~\ref{thm:main}.

\section{Transformer components}\label{app:components}

This appendix specifies the transformer of Theorem~\ref{thm:main} / Proposition~\ref{prop:queue} in full. Throughout, $M = (\Sigma, Q, \delta)$ is the congruent synchronous $K$-queue machine of Proposition~\ref{prop:queue}, with a fixed alphabet $\Sigma$, $|Q| = c$, halting states normalized as absorbing (Section~\ref{sec:models}), and lengths $\ell_r = L(w_r - 1) + 1$ with integers $w_r \ge 2$; $w = \max_r w_r$, $s_{\max} = \max_r \ell_r$, so $wL = s_{\max} - 1 + L$. The model is the looped transformer of Definition~\ref{def:looped} with one decoder layer ($p = 1$), $H = K$ heads (head $r$ serves queue $r$), key dimension $d_k = K$, value dimension $d_v = |\Sigma|$, and loop count $L$; the same block, embedding, and output head are reused at every iteration and every pass. Fix an indexing map $\iota:\Sigma\to\{1,\dots,|\Sigma|\}$.

\paragraph{Indexing conventions.} Coordinates of the residual stream are numbered from $0$, as in the table of \ref{appsub:vocab}: coordinate $0$ holds the constant, the state block occupies coordinates $1,\dots,c$ with state $q$ (states counted from $0$) at coordinate $1+q$, and each symbol or front block $r\in\{0,\dots,K-1\}$ occupies $|\Sigma|$ consecutive coordinates whose within-block offsets are the $1$-based positions $\iota(\alpha)\in\{1,\dots,|\Sigma|\}$ of symbols $\alpha$. Coordinates of value space ($\mathbb{R}^{d_v}$) and the concatenated head outputs ($\mathbb{R}^{|\Sigma|K}$) are numbered from $1$, while coordinates of key space ($\mathbb{R}^{d_k}$) are indexed by the head number $r\in\{0,\dots,K-1\}$: $e_r$ denotes the standard basis vector with index $r$, and $e_0$ exists. Matrix entries are written $W[\text{output coordinate},\,\text{input coordinate}]$, using the conventions of the respective spaces. The attention projections take the residual stream as input: thus $W_q^{(r)}[r,0]=1$ reads the constant (residual coordinate $0$) into key coordinate $r$, while $W_v^{(r)}[\iota(\alpha),\,c+|\Sigma|r+\iota(\alpha)]=1$ reads symbol block $r$, offset $\iota(\alpha)$, into value coordinate $\iota(\alpha)$.

\subsection{Vocabulary and embedding}\label{appsub:vocab}

$\mathcal{V} = (Q \times \Sigma^K) \cup \{\mathrm{ACCEPT}, \mathrm{REJECT}\}$. The generated token at position $i$ is $(q_{Li}, \sigma_{Li})$ --- the state and the $K$-tuple appended at log index $Li$ --- or, at halting, one of the two terminal tokens (\ref{appsub:output}).

$d = 2|\Sigma|K + c + 1$, laid out in blocks:

\begin{center}
\begin{tabular}{llp{0.37\linewidth}}
\toprule
coordinates & block & content \\
\midrule
$0$ & constant & $1$ \\
$1, \dots, c$ & state & one-hot $e_q$ of the state $q$ \\
$c{+}|\Sigma|r{+}1, \dots, c{+}|\Sigma|(r{+}1)$ & symbol $r$ & one-hot of $\sigma^{(r)}\in\Sigma$ \\
$c{+}|\Sigma|K{+}|\Sigma|r{+}1,$ \newline $\dots,\ c{+}|\Sigma|K{+}|\Sigma|(r{+}1)$ & front $r$ & one-hot of the attention-retrieved front; $0$ at embedding \\
\bottomrule
\end{tabular}
\end{center}

$\mathrm{emb}(q, \sigma)$ equals the displayed row of the table above,
\[
\mathrm{emb}(q, \sigma) = \bigl(1,\ e_q,\ \mathrm{onehot}(\sigma^{(0)}), \dots, \mathrm{onehot}(\sigma^{(K-1)}),\ 0, \dots, 0\bigr);
\]
terminal tokens take $\mathrm{emb}(\mathrm{ACCEPT}) = \mathrm{emb}(q_{\mathrm{acc}}, \#^K)$ and $\mathrm{emb}(\mathrm{REJECT}) = \mathrm{emb}(q_{\mathrm{rej}}, \#^K)$ (generation stops at either).

\subsection{Relative positional encoding}\label{appsub:pe}

The relative positional encoding (the form of \citealp{shaw2018self}; \citealp[Eq.~(2)]{liwang2025efficient}): for distance $i-j\in\{0,\dots,w-1\}$,
\[
\mathrm{pos}(i{-}j) \;=\; \Delta\sum_{r:\, i-j = w_r-1} e_r \;\in\; \mathbb{R}^{d_k},
\]
where $\Delta>0$ is the fixed constant of Section~\ref{sec:positional} and $e_r$ is the basis vector of head coordinate $r$. Head $r$ retrieves offset $w_r-1$ by the query $e_r$, so $\langle e_r,\,\mathrm{pos}(i-j)\rangle = \Delta[\,i-j=w_r-1\,]$. The encoding is unlearnable and depends on $(\ell_0,\dots,\ell_{K-1},L)$ only --- the learnable parameter vector is one and the same for all input lengths, the analogue of \citep[Remark~2]{liwang2025efficient}.

\subsection{The attention heads}\label{appsub:heads}

For head $r \in \{0, \dots, K-1\}$, the matrices $W_q^{(r)}, W_k^{(r)} \in \mathbb{R}^{d_k \times d}$ and $W_v^{(r)} \in \mathbb{R}^{d_v \times d}$:

\begin{itemize}
\item \textbf{Query} $W_q^{(r)}$: single nonzero entry $W_q^{(r)}[r, 0] = 1$, copying the constant coordinate of $h$ onto key coordinate $r$. Then $W_q^{(r)}\, h_i^{(l)} = e_r$ for every token, every iteration.
\item \textbf{Key} $W_k^{(r)} = 0$: the keys carry no content; the positional encoding $\mathrm{pos}(i-j)$ of \ref{appsub:pe} enters the key side of the scores (Section~\ref{sec:model}). In the cache only the value entries are informative.
\item \textbf{Score}: $\mathrm{score}_r(i,j) = \langle W_q^{(r)} h_i^{(l)},\ W_k^{(r)} h_j^{(l)} + \mathrm{pos}(i-j)\rangle = \langle e_r,\,\mathrm{pos}(i-j)\rangle = \Delta[\, i - j = w_r - 1\,]$. Over the window $\{i-w+1,\dots,i\}$ this is $\Delta$ at the single position $j^*=i-w_r+1$ (in range at every pass, since $i\ge w-1$ and $w_r\le w$) and $0$ elsewhere. Self-attention has score $0$ because $w_r\ge2$. The unique maximum has gap $\Delta$ (Section~\ref{sec:positional}), so hardmax returns the one-hot vector $s^{(l)}_{r,i,\cdot}=\mathbf{1}_{j^*}$.
\item \textbf{Value} $W_v^{(r)}$: $W_v^{(r)}[\,\iota(\alpha),\ c+|\Sigma|r+\iota(\alpha)\,]=1$ for $\alpha\in\Sigma$ --- the projection onto symbol block $r$. With the write-time invariant (Lemma~\ref{lem:invariant}), $W_v^{(r)}\, h_j^{(l)} = \mathrm{onehot}\big(\sigma^{(r)}_{Lj + l}\big)$, hence the head output
\[
a^{(l)}_{i,r} \;=\; \mathrm{onehot}\Big(\sigma^{(r)}_{L(i - w_r + 1) + l}\Big) \;=\; \mathrm{onehot}\Big(\sigma^{(r)}_{Li + l - \ell_r + 1}\Big),
\]
the front of queue $r$ at machine step $Li + l$ --- the retrieval identity, using $\ell_r - 1 = L(w_r - 1)$.
\end{itemize}

\subsection{Head combination}\label{appsub:combination}

$\mathrm{concat}_r\, a^{(l)}_{i,r} \in \mathbb{R}^{|\Sigma|K}$; $W_O \in \mathbb{R}^{d \times |\Sigma|K}$ with $W_O[\,c+|\Sigma|K+|\Sigma|r+\iota(\alpha),\ |\Sigma|r+\iota(\alpha)\,]=1$ for every $r$ and $\alpha\in\Sigma$, all other entries $0$; $b = 0$. The residual update $\tilde{h}_i^{(l)} = h_i^{(l)} + W_O\,\mathrm{concat}$ writes the $K$ retrieved fronts into the front blocks; the constant, state, and symbol blocks pass through unchanged.

\subsection{Feed-forward network: the transition}\label{appsub:ff}

$\mathrm{FF}$ is a two-layer ReLU network of width $|Q| \cdot |\Sigma|^K + O(d)$ (a machine constant) realizing $\mathrm{FF}(\tilde{h}) + \tilde{h} = (1,\ e_{q'},\ \mathrm{onehot}(\sigma'),\ 0, \dots, 0)$ where $(q', \sigma') = \delta(q, f)$, with $q$ in the state block, $\sigma$ in the symbol blocks, and $f$ in the front blocks of $\tilde h$. For the transition lookup, define
\[
u_{(\hat{q}, \hat{f})}(\tilde{h}) \;=\; 2\,\mathrm{ReLU}\Big(\tilde{h}_{1+\hat{q}} + \textstyle\sum_{r} \tilde{h}_{c+|\Sigma|K+|\Sigma|r+\iota(\hat{f}^{(r)})} - (K + 1) + \tfrac12\Big) \;\in\; \{0, 1\},
\]
equal to $1$ exactly on the table row $(\hat{q}, \hat{f})$ matching $\tilde{h}$; then
\[
\begin{aligned}
\mathrm{FF}(\tilde{h}) \;={}&
\sum_{(\hat{q}, \hat{f})} u_{(\hat{q}, \hat{f})}(\tilde{h}) \cdot
\big(0,\ e_{\delta_1(\hat{q},\hat{f})},\ \mathrm{onehot}(\delta_2(\hat{q},\hat{f})),\ 0^{|\Sigma|K}\big)\\
&-\big(0,\ e_q,\ \mathrm{onehot}(\sigma),\ \mathrm{onehot}(f)\big),
\end{aligned}
\]
where $\mathrm{onehot}(f)$ is the concatenation of the $K$ one-hot front blocks. The linear subtraction is implemented inside the two-layer ReLU network by $O(d)$ passthrough units: every relevant coordinate of $\tilde h$ is nonnegative on the invariant domain, so $\mathrm{ReLU}(\tilde h_a)=\tilde h_a$. Coordinate $0$ is left unchanged. Weights and biases belong to a fixed machine-dependent set, including the threshold bias $-(K+\tfrac12)$; total parameter bit-size is $O((|Q||\Sigma|^K+d)d)$, constant for a fixed machine and independent of $n$ and $L$. On invariant states, all intermediate sums range over a fixed finite set of half-integers, so a fixed machine-dependent numerical format suffices for exact evaluation. The absorbing convention for halting states (Section~\ref{sec:models}) is part of $\delta$; mid-pass halting needs no extra machinery.

\subsection{Output head}\label{appsub:output}

$W_{\mathrm{out}} \in \mathbb{R}^{|\mathcal{V}| \times d}$. For each $(q, \sigma) \in Q \times \Sigma^K$: $W_{\mathrm{out}}[(q,\sigma),\, 1 + q] = 1$ and $W_{\mathrm{out}}[(q,\sigma),\, c+|\Sigma|r+\iota(\sigma^{(r)})] = 1$ for every $r$; terminal rows $W_{\mathrm{out}}[\mathrm{ACCEPT},\, 1 + q_{\mathrm{acc}}] = W_{\mathrm{out}}[\mathrm{REJECT},\, 1 + q_{\mathrm{rej}}] = K + 2$. Since $h_i^{(L)}$ has exactly one nonzero per one-hot block, the ordinary row of the emitted token scores $K + 1$, every other ordinary row less, and the terminal rows score $K + 2$ exactly when the state block is halting (and $0$ otherwise): the argmax is $(q_{L(i+1)}, \sigma_{L(i+1)})$ at non-halting passes and $\mathrm{ACCEPT}$/$\mathrm{REJECT}$ at halting ones. Entries are $O(K)$, constant for a fixed machine.

\subsection{The loop and the input interface}\label{appsub:loop}

One layer per unit (Definition~\ref{def:looped} with $p = 1$): $h_i^{(l+1)} = \mathrm{FF}(\tilde{h}_i^{(l)}) + \tilde{h}_i^{(l)}$ for $l = 0, \dots, L-1$, the matrices of \ref{appsub:vocab}--\ref{appsub:output} identical at every iteration. The per-iteration KV cache (Section~\ref{sec:model}) holds, for each position $j$ and iteration $l$, the value entries $W_v^{(r)} h_j^{(l)} = \mathrm{onehot}(\sigma^{(r)}_{Lj+l})$, $r = 0, \dots, K-1$ --- the keys are zero (\ref{appsub:heads}), the positional encoding entering only the scores (\ref{appsub:pe}). The input interface is the initialization of Section~\ref{sec:initialization}: the value entries of positions $0, \dots, w-2$ at all iterations contain the $K$-tuples $\sigma_0, \dots, \sigma_{s_{\max}-2}$ of the log, and the current token at position $w-1$ is $(q_{s_{\max}-1}, \sigma_{s_{\max}-1})$. Generation starts at pass $w-1$.

\subsection{Faithfulness}\label{appsub:faithful}

\begin{proposition}\label{prop:faithful}
Under the initialization of Section~\ref{sec:initialization}, the transformer of \ref{appsub:vocab}--\ref{appsub:loop} emits, at each pass $i \ge w-1$, the token $(q_{L(i+1)}, \sigma_{L(i+1)})$, and after $\lceil \tau(n)/L \rceil$ passes the terminal token matching $M$'s verdict on $x$; there is no other behavior.
\end{proposition}

\begin{proof}
The proof of Lemma~\ref{lem:invariant} (Section~\ref{sec:initialization}) rests on two structural facts that hold verbatim for the components above: (a) within a pass, state moves only through the residual stream --- the attention reads same-iteration value entries only, and each iteration's feed-forward network consumes the fronts its attention wrote; (b) value entries are frozen at write time. With the base case supplied by the initialization of Section~\ref{sec:initialization} and the retrieval identity of \ref{appsub:heads}, the two-stage induction of Section~\ref{sec:initialization} goes through unchanged; the output head of \ref{appsub:output} turns the final residual into the token or the verdict, and the absorbing convention makes mid-pass halting consistent.
\end{proof}

\section{The exact softmax kernel contract}\label{app:softmax}

This appendix records the kernel contract of Section~\ref{sec:softmax} exactly, the arithmetic behind the $-256$ threshold, and the two kernel variants. Every attention score is a finite binary32 value (IEEE single precision; $\mathrm{RN}_{32}$ denotes round-to-nearest in binary32). Given a score row $(a_j)_j$ and value vectors $(v_j)_j$, the kernel computes the attention output as follows:

\begin{itemize}
\item (i) compute the row maximum $\mu = \max_j a_j$;
\item (ii) form the shifted inputs $\hat a_j = \mathrm{RN}_{32}(a_j - \mu)$;
\item (iii) apply a fixed binary32 exponential primitive $E_{32}$ with $E_{32}(0) = 1$ and $E_{32}(x) = +0$ for every $x \le -256$;
\item (iv) accumulate the numerator $\sum_j E_{32}(\hat a_j)\, v_j$ and the normalizer $\sum_j E_{32}(\hat a_j)$ in binary32, dividing only after both sums;
\item (v) cast the quotient once, at the output boundary, to a fixed format containing $0$ and $1$ exactly.
\end{itemize}

The exactness of Lemma~\ref{lem:softmax} additionally requires the winner's value $v_{j^*}$ to be exactly representable in the accumulation format and preserved by the boundary cast; the fixed format of (v) contains $0$ and $1$ exactly, which covers the one-hot values used in this construction (Corollary~\ref{cor:softmax}), and not arbitrary real value vectors.

\paragraph{The threshold.} Since $256 > 150 \ln 2$, we have $e^{-256} < 2^{-150}$, the round-to-nearest midpoint between $0$ and the least positive binary32 subnormal $2^{-149}$, so a correctly rounded binary32 exponential returns $+0$ on every input at most $-256$. Flush-to-zero only strengthens the property. Any cutoff larger than $150\ln 2$ in magnitude has the same property; $256 = 2^8$ is chosen so that the rescaling factor of Corollary~\ref{cor:softmax} is a power of two.

\begin{remark}\label{rem:kernel}
(1)~The lemma requires a binary32 exponential that returns $+0$ below a fixed threshold, without positive floor or logit soft-capping. (2)~Subtracting the maximum converts overflow into underflow and thereby helps to simulate hard selection. We include it in our softmax kernel because it is standard in production implementations.\footnote{PyTorch's CPU and CUDA softmax kernels (\texttt{host\_softmax} and \texttt{SoftMax.cu}'s epilogue $\exp(\mathrm{input}-\mathrm{max})/\mathrm{sum}$), FlashAttention's online softmax (exponentiating $\mathrm{score}\cdot\mathrm{scale}-\mathrm{max}\cdot\mathrm{scale}$), Triton's fused-softmax kernel (``subtract maximum for numerical stability''), and TensorFlow's softmax op (\texttt{shifted\_logits = logits - max}) all evaluate the exponential only on inputs shifted by the row maximum. This observation concerns that one step; it does not establish full conformance of any production kernel, in particular not of online or fused implementations with their own epilogues.} For a kernel that does not subtract the maximum, the positional encoding of Section~\ref{sec:positional} can set the winner's score to $0$ and every other score to $-256$. (3)~If the kernel includes a positive scale, canonically $1/\sqrt{d_k}$, instantiate $\Delta=\lceil 256\sqrt{d_k}\rceil$. Here $d_k=K$ is a machine constant, so $\Delta$ remains $O(1)$.
\end{remark}

\section{Detailed comparison of Turing-machine simulations}\label{app:comparison}

Table~\ref{tab:comparison} compares transformers simulating $(t, s)$-bounded Turing machines, extending the pointer given in Section~\ref{sec:related}.

\begin{table}[htbp]
\centering
\setlength{\tabcolsep}{2pt}
\begin{tabular}{@{}>{\raggedright\arraybackslash}p{0.195\linewidth}>{\raggedright\arraybackslash}p{0.175\linewidth}>{\raggedright\arraybackslash}p{0.10\linewidth}>{\raggedright\arraybackslash}p{0.09\linewidth}>{\raggedright\arraybackslash}p{0.135\linewidth}>{\raggedright\arraybackslash}p{0.235\linewidth}@{}}
\toprule
Source & Precision & Dim. & Depth & Window & CoT per TM-step \\
\midrule
Pérez et al. 2021 & $O(\log t)$ & $O(1)$ & $O(1)$ & $n + t$ & $1$ \\
Bhattamishra et al. 2020 & unbounded & $O(1)$ & $1$ & $n + t$ & $1$ \\
Merrill \& Sabharwal 2024 & $O(\log t)$ & $O(1)$ & $O(1)$ & $n + t$ & $1$ \\
Li et al. 2024 & $O(1)$ & $O(\log t)^{\ddagger}$ & $2$ & $O(t\log t)$ & $O(\log t)$ amort. \\
Qiu et al. 2025 & $O(\log t)$ & $O(1)$ & $O(1)$ & $O(t\log t)$ & $O(\log t)$ amort. \\
Li \& Wang 2025a & $O(1)$ & $O(1)$ & $1$ & $s$ & $s$ \\
Li \& Wang 2025b & $O(1)$ & $O(1)$ & ---$^{\dagger}$ & $s{+}1$ & $O(s^{1/k})$ \\
Jiang et al. 2026$^{\diamond}$ & unbounded & $O(1)$ & $O(1)$ & $n + 2^{O(t)}$ & $1$ per counter step \\
Br\"osamle \& Eckstein 2026$^{\diamond}$ & $O(\log\log s)$ & $O(\log s)$ & $O(\log s)$ & $O(s)$ (SCoT) & $O(1)$ amort. \\
\textbf{This work}$^{\diamond}$ & $O(1)$ & $O(1)$ & $L$ & $O(s/L)$ & $O(s^{1/k}/L)$ \\
\bottomrule
\end{tabular}
\vspace{2pt}
\begin{minipage}{\linewidth}
\footnotesize
$^{\diamond}$ marks constructions using softmax attention (this work via the finite-precision kernel of Section~\ref{sec:softmax}).

$^{\dagger}$ Li \& Wang 2025b report only the head--layer product $K = H{\cdot}L$; their depth is left unspecified.

$^{\ddagger}$ \citet{li2024chain} state $O(\log n)$ embedding dimension under the assumption that $t(n)$ is polynomial, so $\lceil\log_2(t(n)+n)\rceil=O(\log n)$.
\end{minipage}
\caption{Comparing transformers simulating $(t, s)$-bounded Turing machines. Depth is the number of sequential decoder layers as reported by each source; for this work it is the loop count $L$---one shared decoder block executed $L$ times per generated token. This work prescribes $L$ and chooses $k$ subject to $s(n)^{1/k}=\Omega(L)$; Li \& Wang 2025b is the endpoint $L{=}1$.}
\label{tab:comparison}
\end{table}

\paragraph{Notation normalization.} The rows ``Li \& Wang 2025a'' and ``Li \& Wang 2025b'' abbreviate \citet{liwang2025constant} and \citet{liwang2025efficient}, respectively; the row ``Jiang et al.\ 2026'' carries the venue year of \citet{jiang2025softmax}. Depth is printed as $L$ for this work, where the source of each competitor row uses its own layer-count symbol; the main text likewise uses $L$ for the generic decoder depth, which coincides with the loop count for the one-layer loop unit. This work's window and CoT bounds presuppose the admissibility condition $s(n)^{1/k}=\Omega(L)$ and the loaded-cache interface of Section~\ref{sec:initialization}; the initialization cost is excluded from this work's bounds (the comparison rows report the sources' own generation accounting).

\section{Notation}\label{app:notation}

\begin{center}
\begin{tabular}{@{}llp{0.52\linewidth}@{}}
\toprule
Symbol & Formally defined & Meaning \\
\midrule
$t(n),\,s(n)$ & \ref{sec:models} & time and space bounds of the simulated Turing machine \\
$k,\,k'$ & \ref{sec:models} & fixed hierarchy exponent ($s(n)^{1/k}=\Omega(L)$); tape count \\
$K$ & \ref{fact:queue}, \ref{app:congruence} & queue count $=8kk'$; number of attention heads \\
$L$ & \ref{def:looped} & loop count, and the decoder depth: one shared block executed $L$ times per token \\
$m$ & \ref{def:log} & execution-log step index \\
$g$ & \ref{fact:queue} & hierarchy level index ($1\le g\le k$) \\
$(\sigma_m,q_m)$ & \ref{def:log} & execution log: $K$-tuple appended and state at step $m$ \\
$\tau(n)$ & \ref{def:log} & run length after the cache is filled \\
$T_{\mathrm{cot}}$ & \ref{prop:queue} & generated tokens; exactly $\lceil\tau(n)/L\rceil$, asymptotically $O(t(n)s(n)^{1/k}/L)$ \\
$i,\,l$ & \ref{def:looped}, \ref{sec:hardmax} & pass (position) and iteration indices \\
$\ell_r,\,w_r,\,w$ & \ref{def:log}, \ref{app:congruence} & length of queue $r$; $w_r=(\ell_r-1)/L+1$; $w=\max_r w_r$ \\
$s_{\max}$ & \ref{def:log} & $\max_r\ell_r$ \\
$b,\,B$ & \ref{fact:queue}, \ref{app:congruence} & congruent queue base; $B=\lceil s(n)^{1/k}\rceil$ (cited conversion) \\
$d,\,d_k,\,d_v,\,c$ & \ref{sec:tokens}, \ref{appsub:vocab} & residual dimension $2|\Sigma|K+|Q|+1$; key dim.\ $K$; value dim.\ $|\Sigma|$; $|Q|$ \\
$a_{i,j}$ & \ref{sec:softmax} & attention score (row $i$, candidate $j$) \\
$\Delta$ & \ref{sec:positional} & score gap; set to $256$ in Corollary~\ref{cor:softmax} \\
$\mu,\,\hat a_j,\,E_{32}$ & \ref{app:softmax} & kernel row maximum; shifted input; binary32 exponential \\
$p$ & \ref{def:looped} & layers per loop unit ($p=1$ in the construction) \\
\bottomrule
\end{tabular}
\end{center}

\end{document}